\documentclass[10pt,twocolumn,conference]{IEEEtran}

\newcommand{\beq}{\begin{equation}}
\newcommand{\enq}{\end{equation}}
\newcommand{\beqa}{\begin{eqnarray}}
\newcommand{\enqa}{\end{eqnarray}}

\newcommand{\be}{\beta}

\newcommand{\qed}{\hfill $\Box$}

\newtheorem{theorem}{Theorem}
\newenvironment{proof}{{\sl Proof\/}:\ \ }{\qed\vspace{\baselineskip}}

\def\spacing#1{\renewcommand{\baselinestretch}{#1}\large\normalsize}
\def\bbC{{\sf C}\kern -6pt {\sf C}}
\def\bbF{{\sf F}\kern -5pt {\sf F}}
\def\bbR{{\sf R}\kern -6pt {\sf R}}
\def\bbZ{{\sf Z}\kern -5pt {\sf Z}}
\def\sfbegin{\begingroup\sf}
\def\sfend{\endgroup}

\def\be{\begin{eqnarray*}}
\def\ee{\end{eqnarray*}}

\usepackage[final]{graphicx}
\usepackage[reqno]{amsmath}
\usepackage{amssymb}
\usepackage{subfig}
\usepackage{epstopdf,comment,color}
\IEEEoverridecommandlockouts

\title{Are Slepian-Wolf Rates Necessary for Distributed Parameter Estimation?}
\author{Mostafa El Gamal and Lifeng Lai\\Department of Electrical and Computer Engineering\\Worcester Polytechnic Institute\\\{melgamal, llai\}@wpi.edu\thanks{The work of Mostafa El Gamal and Lifeng Lai was supported by the Qatar National
Research Fund under Grant QNRF-6-1326-2-532. }}
\date{ }

\begin{document}
\maketitle \spacing{1} 

\begin{spacing}{1}

\begin{abstract}
We consider a distributed parameter estimation problem, in which multiple terminals send messages related to their local observations using limited rates to a fusion center who will obtain an estimate of a parameter related to observations of all terminals. It is well known that if the transmission rates are in the Slepian-Wolf region, the fusion center can fully recover all observations and hence can construct an estimator having the same performance as that of the centralized case. One natural question is whether Slepian-Wolf rates are necessary to achieve the same estimation performance as that of the centralized case. In this paper, we show that the answer to this question is negative. We establish our result by explicitly constructing an asymptotically minimum variance unbiased estimator (MVUE) that has the same performance as that of the optimal estimator in the centralized case while requiring information rates less than the conditions required in the Slepian-Wolf rate region.
\end{abstract}


\section{Introduction}\label{sec:intro}
%
There are two main different setups for statistical learning: centralized learning and distributed learning. In the centralized learning, which has been studied extensively, all data is available at a centralized location. In the distributed learning, data is stored in multiple terminals. The distributed learning setup has attracted significant recent research interests as the data involved in learning is increasingly large in volume and might be stored in multiple terminals~\cite{maxim,aolin,jordan,shamir}. For the distributed learning, each terminal either has a few observations about all variables, or has full knowledge about a subset of variables (all observations about a subset of variables). The first scenario is relatively easier since each terminal can still make its own local inference without even communicating with each other, while communication between terminals is essential for the second scenario. In this paper, we focus on the more challenging second scenario.

In particular, we consider a distributed parameter estimation problem. In the setup considered, there are two random variables $(X,Y)$ with a joint probability mass function (PMF) $P_{\theta}(X,Y)$ parameterized by an unknown parameter $\theta$. Two terminals $A$ and $B$ observe $X^n$ and $Y^n$ respectively and send messages related to their own local observations with limited rates to terminal $C$, which will then obtain an estimate of the unknown parameter. It is well known that if the transmission rates from the terminals are inside the Slepian-Wolf rate region~\cite{slepian}, there exists a universal coding scheme~\cite{csiszar} that enables terminal $C$ to fully recover $(X^n,Y^n)$. Hence, once the transmission rates are inside the Slepian-Wolf rate region, the performance of the best estimator for the distributed setup is the same as that of the best estimator for the centralized case.

One natural question is: are Slepian-Wolf rates \emph{necessary} to achieve the same estimation performance as that of the centralized case? The answer to this question has significant implications in the distributed estimation. If the answer is yes, then to obtain the best estimate of the unknown parameter requires transmission rates to be so high that they are sufficient to fully recover the observations at the decoder, hence no rate reduction is possible. On the other hand, if the answer is no, then the observations can be compressed beyond the limits of source coding for full observation recovery. At a first glance, the answer to this question should be no as we are only interested in estimating a parameter related to the observations and are not interested in recovering the observations themselves. However, all existing related works indicate otherwise. For example, \cite{zia} addressed the same question and suggested that Slepian-Wolf rates might be necessary. In addition, the performance of the best known estimator by Han and Amari~\cite{han} does not match that of the centralized case when the information rates are outside of the Slepian-Wolf rate region. Furthermore,~\cite{ahlswede} showed that, under certain conditions, extracting even one bit of information from distributed sources is as hard as recovering full observations and hence requires the information rates to be in the Slepian-Wolf rate region.

In this paper, we show that the answer to this question is indeed \emph{no}. We establish our result by explicitly constructing a distributed estimator that achieves the same performance as that of the optimal estimator for the centralized case while using information rates outside of the Slepian-Wolf region. In particular, we consider binary symmetric sources (i.e., both $X^n$ and $Y^n$ are binary sequences) parameterized by an unknown parameter $\theta$. In our scheme, we first design a universal coding/decoding scheme that enables terminal $C$ to compute $Z^n=X^n\oplus Y^n$, which can be achieved using rates outside of the Slepian-Wolf rate region, and then construct an estimator using $Z^n$. We show that our estimator is an asymptotically minimum variance unbiased estimator (MVUE) \cite{poor} and achieves the same variance index as that of the best estimator in the centralized case. We further extend our scheme to a more general class of joint PMFs and show that our scheme can also achieve the same performance as that of the best estimator in the centralized case while using transmission rates less than the conditions required in the Slepian-Wolf rate region. The key idea of our scheme is, instead of fully recovering the source observations, we aim to recover sufficient statistics at terminal $C$ using less information rates.

The rest of the paper is organized as follows. We introduce the problem formulation in Section~\ref{sec:formulation}. In Section~\ref{sec:binary}, we establish our main results for the binary symmetric sources. We extend our work to a more general class of information sources in Section~\ref{sec:general}. We present the simulation results in Section~\ref{sec:simulation}. Finally, we conclude the paper in Section~\ref{sec:conclusion}.

\section{Problem Formulation} \label{sec:formulation}
Consider two information sources $X$ and $Y$ taking values from the discrete alphabets $\mathcal X$ and $\mathcal Y$, respectively. $(X^n,Y^n)=\{(X_i,Y_i)\}_{i=1}^n$ are $n$ independently and identically distributed (i.i.d.) observations drawn according to the parametric joint PMF $P_\theta(X,Y)$ where $\theta \in \Theta$ is the unknown parameter. We consider a distributed setup in which $X^n$ are observed at terminal $A$ and $Y^n$ are observed at terminal $B$. Using limited rates, these two terminals send messages related to their own local observations to a fusion center (terminal $C$), which will then obtain an estimate $\hat \theta$ of $\theta$ using these messages. The setup is illustrated in Fig.~\ref{Fig-1}.

\begin{figure}[htb]
\centering
\includegraphics[width=0.4\textwidth]{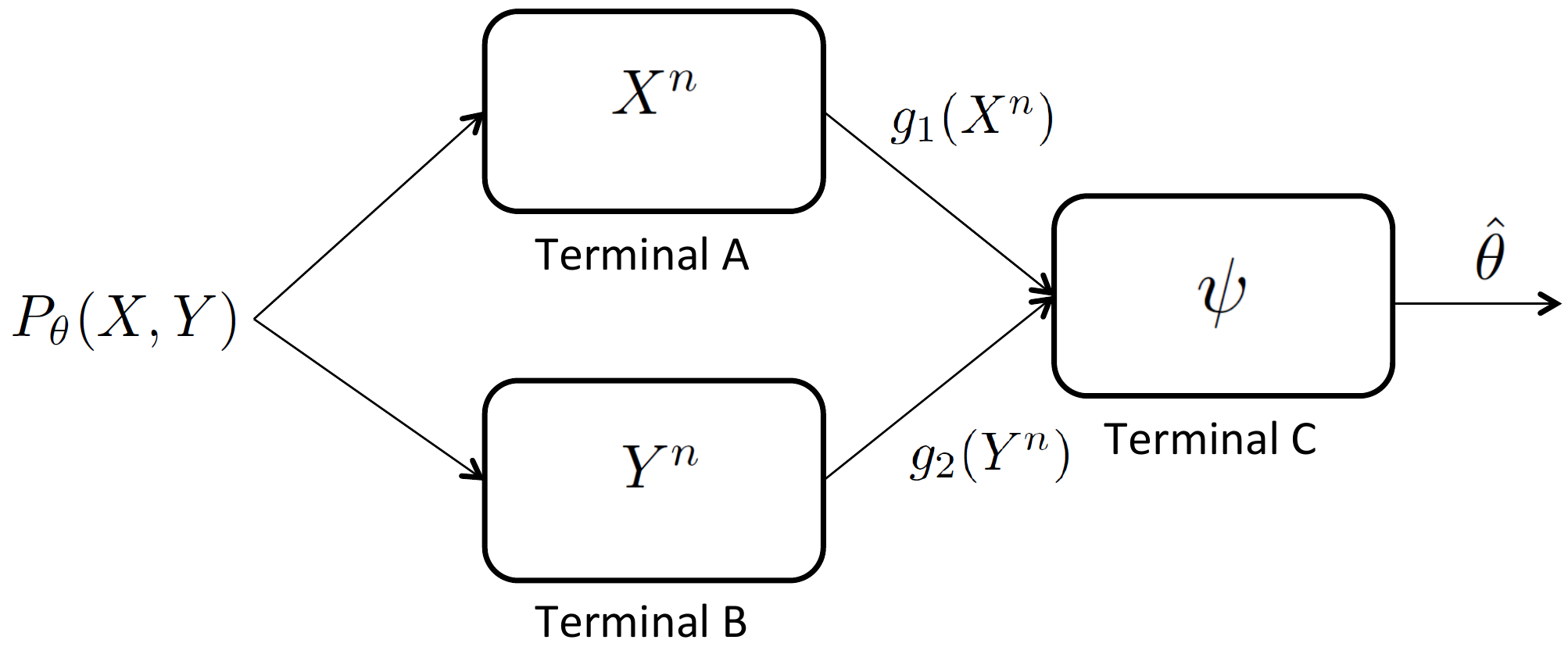}\\
\caption{System Model.}
\label{Fig-1}
\end{figure}

In particular, terminal $A$ employs an encoding function $g_1: X^n \to g_1(X^n)$, while terminal $B$ employs an encoding function $g_2: Y^n \to g_2(Y^n)$. The code rates are
\begin{eqnarray}
R_X= \frac{\log {||g_1||}}{n}, R_Y=  \frac{\log {||g_2||}}{n},
\end{eqnarray}
where $||g_i||$ is the cardinality of the encoding function $g_i$.

From $g_1(X^n)$ and $g_2(Y^n)$, the decoder obtains an estimate $\hat{\theta}$ of the unknown parameter $\theta$ using estimator $\psi$:
\begin{eqnarray}
\hat \theta = \psi(g_1(X^n),g_2(Y^n)).
\end{eqnarray}

To evaluate the quality of the estimator, we use the variance index that is defined as
\begin{eqnarray}
V[\hat \theta]=\lim_{n \rightarrow \infty} n\text{Var}_{\theta}[\hat \theta] = \lim_{n \rightarrow \infty} n\mathbb E_{\theta}[(\hat \theta - \mathbb E[\hat \theta])^2].
\end{eqnarray}
It is desirable to have an estimator that is asymptotically unbiased, i.e., $\mathbb E_{\theta}[\hat{\theta}] \to \theta$ as $n \to \infty$, and has a small variance index.

It is well-known that, if the coding rates satisfy (will be called Slepian-Wolf rates in the sequel)
\begin{eqnarray}
R_X \ge H_{\theta}(X|Y), \label{eq:ratex}\\
R_Y \ge H_{\theta}(Y|X), \label{eq:ratey}\\
R_X+R_Y \ge H_{\theta}(X,Y),\label{eq:ratexy}
\end{eqnarray}
there exists universal source coding schemes~\cite{csiszar} (i.e., the coding scheme does not depends on the value of the unknown parameter $\theta$) such that the decoder can reconstruct $X^n$ and $Y^n$ with a diminishing error probability. Here, $H_{\theta}(\cdot)$ and $H_{\theta}(\cdot|\cdot)$ denote the entropy and conditional entropy respectively\footnote{Throughout the paper, we use the subscript $\theta$ to emphasize the fact that value of the quantity of interest depends on the parameter $\theta$.}. Hence, if~\eqref{eq:ratex}-\eqref{eq:ratexy} are satisfied, we can obtain the same estimation performance as that of the centralized case. 

The question we ask in this paper is: are Slepian-Wolf rates \emph{necessary} to achieve the same estimation performance as that of the centralized case?~\cite{zia} investigated the same question and suggested that Slepian-Wolf rates appear to be necessary for achieving the centralized estimation performance. In this paper, we show that Slepian-Wolf rates are \emph{not} necessary. In particular, we show that there indeed exists a class of PMFs and the corresponding distributed estimators that require communication rates less than the Slepian-Wolf rates while still achieving the same performance as that of the best estimator for the centralized case.

Throughout the paper, we use an upper case letter $Z$ to denote a random variable, and a lower case letter $z$ to denote a realization of $Z$. For any sequence $z^n=(z(1),\cdots,z(n)) \in \mathcal{Z}^n$, the relative frequencies (empirical PMF) $\pi(a|z^n)\triangleq n(a|z^n)/n,\forall a\in\mathcal{Z}$ of the components of $z^n$ is called the type of $z^n$. 
Here $n(a|z^n)$ is the total number of indices $t$ at which $z(t)=a$.  

\section{Binary Symmetric Case} \label{sec:binary}
In this section, we consider the case of binary symmetric sources with $|\mathcal{X}|=|\mathcal{Y}|=2$ and a joint PMF of $(X,Y)$ as given in Table~\ref{bsc}, in which the unknown parameter $\theta \in \Theta = (0,1)$.
\begin{table}[htb] 
\centering
\begin{tabular}{l | c | c}
$X/Y$ & 0 & 1\\
\hline \hline
0 & $ \theta/ 2$ & $(1-\theta)/2$\\
1 & $(1-\theta)/2$ &  $\theta/ 2$\\
\end{tabular}
\caption{The joint PMF of binary symmetric sources.}
\label{bsc}
\end{table}

We show that, to estimate $\theta$ for this class of PMFs, we can achieve the centralized estimation performance using rates that do not satisfy~\eqref{eq:ratex}-\eqref{eq:ratexy}. We establish this result using two steps: 1) in the first step, we design a universal encoder at terminals $A$ and $B$ and universal decoder at terminal $C$ to compute the modulo-two sum $Z^n= X^n \oplus Y^n$; 2) in the second step, we construct an estimator using $Z^n$.


\subsection{Step 1: Comupting $Z^n$}\label{sub:encoding}
Here, we discuss how to universally compute $Z^n=X^n\oplus Y^n$ at terminal $C$. Towards this goal, we will use the same linear code at both encoders and use a minimum entropy decoder at terminal $C$.

Since the encoders at terminals $A$ and $B$ are the same, we use the following simplified notation
\begin{eqnarray}
f&=& g_1= g_2, \nonumber \\
R&=& R_X= R_Y \label{5}.
\end{eqnarray}

The following theorem shows that as long as $R\geq H_{\theta}(X|Y)=H_{\theta}(Y|X)$, the decoder can reconstruct $Z^n$ with a diminishing error probability.

\begin{theorem}\label{thm:decodingZ}
If 
\begin{equation}
R \ge H_{\theta}(X|Y)=H_{\theta}(Y|X),
\end{equation}
there exist universal encoding/decoding functions to reconstruct $Z^n=X^n\oplus Y^n$ at terminal $C$ with an exponentially decreasing error probability.
\end{theorem}

\begin{proof}
The proof follows a similar structure as the proofs in \cite{korner} and \cite{csiszar}. In particular, using the ideas in~\cite{csiszar}, we modify the proof of~\cite{korner} to make it universal. 


\noindent\textbf{Random Code Generation}: We use a linear code $f$ with an encoding matrix $A$ of size $n\times nR$ to map $\{0,1\}^n$ to $\{1,2,...,2^{nR}\}$. Hence $||f||=2^{nR}$. We independently generate each entry of $A$ using a uniform binary distribution, i.e., each entry of $A$ is $0$ or $1$ with probability $0.5$.

\noindent\textbf{Encoding}:
The encoded messages of the realizations $x^n \in \{0,1\}^n$ and $y^n \in \{0,1\}^n$ are
\begin{eqnarray}
f(x^n)&=&x^nA, \nonumber \\
f(y^n)&=&y^nA.
\end{eqnarray}

\noindent\textbf{Decoding}:
The decoder first combines the messages into a single message as
\begin{equation}
f(x^n) \oplus f(y^n),
\end{equation}
in which $\oplus$ denotes the element-wise xor.

It follows from the code linearity that
\begin{equation}
 f(x^n) \oplus f(y^n)=f(x^n \oplus y^n)=f(z^n).
\end{equation}

From $f(x^n\oplus y^n)$, terminal $C$ uses a minimum entropy decoder to obtain $\hat{z}^n$. In particular, for each $\bar{z}^n$ such that $f(\bar{z}^n)=f(x^n\oplus y^n)$, the minimum entropy decoder first calculates the entropy of its type, then picks the one that has the least entropy to be the decoded sequence.
In the following, to simplify the notation, we use $\bar{Z}^{(n)}$ and $Z^{(n)}$ to denote dummy random variables whose PMFs $P_{\bar{Z}^{(n)}}$ and $P_{Z^{(n)}}$ are the same as the types of $\bar{z}^n$ and $z^n$, respectively.
The final decoded message is denoted as
\begin{equation}
\hat z^n= \phi(f(z^n)),
\end{equation}
where $\phi$ denotes the minimum entropy decoding function. 

\noindent\textbf{Error Probability Analysis:}
A decoding error occurs if and only if there exists a sequence $\hat z^n \not= z^n$ such that
\begin{equation}
f(\hat z^n)= f(z^n) ~~\text{and}~~ H(\hat Z^{(n)}) \le H(Z^{(n)}).
\end{equation}

The error probability, averaging over all possible codebooks, is
\begin{eqnarray}
P_e^{(n)}= \sum_{z^n\in \{0,1\}^n} P_r(z^n)P_r(\hat z^n \neq z^n).\label{eq:errorprob}
\end{eqnarray}
To analyze the probability of the decoding error, let $\tilde z^n \in \{0,1\}^n$ denote another sequence such that
\begin{equation}
\tilde z^n \neq z^n, ~~~~~f(\tilde z^n) = f(z^n) \label{4}.
\end{equation}

Let $\tilde Z^{(n)}$ be a dummy random variable whose PMF $P_{\tilde Z^{(n)}}$ is the same as the type of $\tilde z^n$. Define $\mathcal P^{(n)}_{Z\tilde Z}$ as the set of all joint types between any two sequences $z^n$ and $\tilde z^n$. For any given $f$ (equivalently for a given encoding matrix $A$), define $N_f^n(Z\tilde Z)$ as the number of sequences $z^n$ such that there exists another sequence $\tilde{z}^n$ having the joint type $P_{Z^{(n)}\tilde Z^{(n)}} \in \mathcal P^{(n)}_{Z\tilde Z}$ and \eqref{4} holds.

Since each entry in $A$ is uniformly distributed, then each element in $f(z^n)$ is uniformly distributed if $z^n$ is a nonzero sequence. Therefore,
\begin{eqnarray}
P_r(f(z^n) = 0)=(0.5)^{nR}= \frac 1 {||f||},
\end{eqnarray}
in which the probability is computed over all codebooks. This implies that
\begin{eqnarray}
P_r(f(\tilde z^n) = f(z^n)) = P_r(f(\tilde z^n - z^n) = 0)= \frac 1 {||f||} \label{6}.
\end{eqnarray}

Define $T_{P_{Z^{(n)} \tilde Z^{(n)}}}$ as the set of all sequence pairs $(z^n,\tilde z^n)$ that have the joint type $P_{Z^{(n)} \tilde Z^{(n)}}$, $T_{P_{Z^{(n)}}}$ as the set of all sequences $z^n$ that have the marginal type $P_{Z^{(n)}}$, and $T_{P_{\tilde Z^{(n)} | Z^{(n)}}}(z^n)$ as the set of all sequences $\tilde z^n$ that have the joint type $P_{Z^{(n)}\tilde Z^{(n)}}$ with $z^n$. The sizes of the sets $T_{P_{Z^{(n)}}}$ and $T_{P_{\tilde Z^{(n)} | Z^{(n)}}}(z^n)$ are bounded as \cite{csiszar2}
\begin{eqnarray}
|T_{P_{Z^{(n)}}}| &\le& 2^{nH(Z^{(n)})},  \nonumber \\
|T_{P_{\tilde Z^{(n)} | Z^{(n)}}}(z^n)| &\le& 2^{nH(\tilde Z^{(n)}| Z^{(n)})+\epsilon} \label{7},
\end{eqnarray}
where $\epsilon$ is an arbitrary small number. Notice that, for any given $P_{Z^{(n)}\tilde Z^{(n)}}$, $N_f^n(Z\tilde Z)$ is a random variable (random over $f$) that can be expressed as
\begin{eqnarray}
&&N_f^n(Z\tilde Z)= \sum_{z^n \in T_{P_{Z^{(n)}}}} {\bf 1}\big(\exists\tilde z^n \not= z^n:f(\tilde z^n) = f(z^n), \nonumber \\
&&\hspace{34mm}\text{and } (z^n,\tilde z^n) \in T_{P_{Z^{(n)}\tilde Z^{(n)}}}\big) \nonumber \\
&&\hspace{14mm}= \sum_{z^n \in T_{P_{Z^{(n)}}}} {\bf 1}\big(\exists\tilde z^n \not= z^n:f(\tilde z^n) = f(z^n), \nonumber \\
&&\hspace{34mm}\text{and } \tilde z^n \in T_{P_{\tilde Z^{(n)} | Z^{(n)}}}(z^n)\big),
\end{eqnarray}
where $ {\bf 1}(\cdot)$ is the indication function. The expectation of $N_f^n(Z\tilde Z)$ over all possible codebooks $f$ is
\begin{eqnarray}
&&\hspace{-4mm}\mathbb E[N_f^n(Z\tilde Z)]\nonumber\\
&&\hspace{-4mm}= \sum_{z^n\in T_{P_{Z^{(n)}}}} \mathbb E\big[{\bf 1}\big(\exists~\tilde z^n \not= z^n~:~f(\tilde z^n) = f(z^n), \nonumber \\
&&\hspace{22mm}\text{and } \tilde z^n \in T_{P_{\tilde Z^{(n)} | Z^{(n)}}}(z^n)\big)\big] \nonumber \\
&&\hspace{-4mm}\le  \sum_{z^n\in T_{P_{Z^{(n)}}}} \sum_{\tilde z^n \in T_{P_{\tilde Z^{(n)} | Z^{(n)}}}(z^n)} P_r(f(\tilde z^n) = f(z^n))  \label{18}.
\end{eqnarray}

(\ref{6}), (\ref{7}), and (\ref{18}) imply that
\begin{equation}
\mathbb E[N_f^n(Z\tilde Z)] \le \frac {2^{n(H(Z^{(n)})+H(\tilde Z^{(n)}|Z^{(n)})+\epsilon)}} {||f||}.
\end{equation}

Applying the Markov's inequality, we have
\begin{eqnarray}
&&\hspace{-5mm}P_r\left(N_f^n(Z\tilde Z) \ge \frac {2^{n(H(Z^{(n)})+H(\tilde Z^{(n)}|Z^{(n)})+\epsilon)}(|\mathcal P^{(n)}_{Z\tilde Z}|+\delta)} {||f||}\right)\nonumber \\
&&\hspace{3mm}\le \frac {1} {|\mathcal P^{(n)}_{Z\tilde Z}|+\delta} \label{10},
\end{eqnarray}
where $|\mathcal P^{(n)}_{Z\tilde Z}|$ is the total number of possible joint types and $\delta$ is an arbitrary small number. To simplify the notation, let
\begin{equation}
B^n(Z \tilde Z) \triangleq \frac {2^{n(H(Z^{(n)})+H(\tilde Z^{(n)}|Z^{(n)})+\epsilon)}(|\mathcal P^{(n)}_{Z\tilde Z}|+\delta)} {||f||}.
\end{equation}

Considering all joint types $P_{Z^{(n)}\tilde Z^{(n)}}$ simultaneously, the union bound and (\ref{10}) imply that
\begin{eqnarray}
&&\hspace{-6mm}P_r\left(N_f^n(Z\tilde Z) \le B^n(Z \tilde Z),~\forall P_{Z^{(n)}\tilde Z^{(n)}} \in \mathcal P^{(n)}_{Z\tilde Z}\right)\nonumber
\\ &\ge& 1- \sum_{1}^{|\mathcal P^{(n)}_{Z\tilde Z}|} \frac 1 {|\mathcal P^{(n)}_{Z\tilde Z}|+\delta} \nonumber
\\ &>& 0\label{22}.
\end{eqnarray}

Since the probability in (\ref{22}) is positive, then there exists a codebook $f^*$ that the following equation holds for all joint types $P_{Z\tilde Z}$ simultaneously
\begin{equation}
N_{f^*}^n(Z\tilde Z) \le \frac {2^{n(H(Z^{(n)})+H(\tilde Z^{(n)}|Z^{(n)})+\epsilon)}(|\mathcal P^{(n)}_{Z\tilde Z}|+\delta)} {||f^*||} \label{11}.
\end{equation}

As $||f^*||=2^{nR}$ and $|\mathcal P^{(n)}_{Z\tilde Z}| \le (n+1)^4$, we further have
\begin{eqnarray} \label{eq:Nbound}
&&N_{f^*}^n(Z\tilde Z)\\\nonumber
 &&\le ((n+1)^4+\delta)~~2^{n(H(Z^{(n)})+H(\tilde Z^{(n)}|Z^{(n)})+\epsilon-R)}.
\end{eqnarray}

In the following, we will focus on $f^*$.

Let $P^{(n)}_{e,f^*}(Z \tilde Z)$ denote the portion of error probability associated with a fixed joint type $P_{Z^{(n)}\tilde Z^{(n)}}$
\begin{eqnarray}
&&P^{(n)}_{e,f^*}(Z \tilde Z)\\
&&\triangleq \sum_{z^n \in T_{P_{Z^{(n)}}}} P_r(z^n) {\bf 1}\big(\exists\tilde z^n \not= z^n:f^*(\tilde z^n) = f^*(z^n), \nonumber \\
&&\hspace{33mm}\text{and }(z^n,\tilde z^n) \in T_{P_{Z^{(n)}\tilde Z^{(n)}}}\big).\nonumber
\end{eqnarray}

The total decoding error probability $P^{(n)}_{e,f^*}$, when using $f^*$, can be expressed as
\begin{equation}
P^{(n)}_{e,f^*} = \sum_{P_{Z^{(n)}\tilde Z^{(n)}}} P^{(n)}_{e,f^*}(Z \tilde Z) \label{error}.
\end{equation}
Let $A_{\epsilon_1}^{(n)}$ denote the set of marginal types $P_{Z^{(n)}}$ such that $|P_{Z^{(n)}}(z=i)-P_\theta(z=i)|< \frac {\epsilon_1} 2$ for $i \in \{0,1\}$, where $\epsilon_1$ is an arbitrarily small number. Using the definition of $A_{\epsilon_1}^{(n)}$, (\ref{error}) can be rewritten as
\begin{eqnarray}
 P^{(n)}_{e,f^*} &=&\sum_{P_{Z^{(n)}\tilde Z^{(n)}}, P_{Z^{(n)}} \in A_{\epsilon_1}^{(n)}}  P^{(n)}_{e,f^*}(Z \tilde Z)
\nonumber\\
&&\hspace{6mm}+ \sum_{P_{Z^{(n)}\tilde Z^{(n)}}, P_{Z^{(n)}} \in \bar{A}_{\epsilon_1}^{(n)}}  P^{(n)}_{e,f^*}(Z \tilde Z) \nonumber \\ \label{error1}
&\triangleq& S_1 + S_2,
\end{eqnarray}
where $\bar{A}_{\epsilon_1}^{(n)}$ denotes the complimentary set of $A_{\epsilon_1}^{(n)}$.
For $S_2$, we have that
\begin{eqnarray}
P^{(n)}_{e,f^*}(Z \tilde Z) \le 2^{-n(D(P_{Z^{(n)}}||P_\theta(Z)))}, \label{error2}
\end{eqnarray}
where $D(P_{Z^{(n)}}||P_\theta(Z))$ is the KL divergence between the marginal type $P_{Z^{(n)}}$ and the true PMF $P_\theta(Z)$ of $Z=X \oplus Y$. Using Pinsker's inequality, for $P_{Z^{(n)}} \in \bar{A}_{\epsilon_1}^{(n)}$, we have
\begin{eqnarray}
D(P_{Z^{(n)}}||P_\theta(Z)) \ge 2\epsilon_1^2.
\end{eqnarray}
Therefore,
\begin{eqnarray}
S_2 &\le& \sum_{P_{Z^{(n)}\tilde Z^{(n)}}} 2^{-2n\epsilon_1^2} \nonumber \\
&\le& (n+1)^4~~2^{-2n\epsilon_1^2}. \label{error3}
\end{eqnarray}
(\ref{error3}) implies that $S_2 \rightarrow 0$ exponentially as $n \rightarrow \infty$.

For $S_1$, we have that
\begin{equation}
P^{(n)}_{e,f^*}(Z \tilde Z) \le N_{f^*}^n(Z\tilde Z)~~2^{-n(H(Z^{(n)})+D(P_{Z^{(n)}}||P_\theta(Z))}.
\end{equation}
Using~\eqref{eq:Nbound}, we further have
\begin{eqnarray}
&&P^{(n)}_{e,f^*}(Z \tilde Z) \le   \\ \nonumber &&((n+1)^4+\delta)~2^{-n\big(D(P_{Z^{(n)}}||P_\theta(Z))+R-H(\tilde Z^{(n)}|Z^{(n)})-\epsilon\big)}.
\end{eqnarray}
As we use the minimum entropy decoder, we have $H(\tilde Z^{(n)}) \le H(Z^{(n)})$, which implies $H(\tilde Z^{(n)}|Z^{(n)})\le H(\tilde Z^{(n)})\le H(Z^{(n)})$. Therefore,
\begin{eqnarray}
&&P^{(n)}_{e,f^*}(Z \tilde Z) \\ \nonumber &&\le ((n+1)^4+\delta)~2^{-n\big(D(P_{Z^{(n)}}||P_\theta(Z))+R-H(Z^{(n)})-\epsilon\big)}.
\end{eqnarray}

Since $P_{Z^{(n)}} \in A_{\epsilon_1}^{(n)}$, it is easy to check that
\begin{eqnarray}
|H(Z^{(n)})-H_{\theta}(Z)| \le D(P_{Z^{(n)}}||P_\theta(Z))+ \epsilon_2.
\end{eqnarray}
Here
\begin{eqnarray}
\epsilon_2=-\frac {\epsilon_1} 2 \sum_i \log P_\theta(z=i),
\end{eqnarray}
which can be made arbitrarily small as $\epsilon_1\downarrow 0$ for $\theta\in(0,1)$.

Therefore,
\begin{eqnarray}
&&\hspace{-6mm}P^{(n)}_{e,f^*}(Z \tilde Z)\\  \nonumber  && \le ((n+1)^4+\delta)~2^{-n\big(R-H_{\theta}(Z)-\epsilon_3\big)},
\end{eqnarray}
in which $\epsilon_3=\epsilon+\epsilon_2$.

This implies that $S_1 \rightarrow 0$ exponentially as $n \rightarrow \infty$ if
\begin{equation}
R \ge H_{\theta}(Z) \label{15}.
\end{equation}

Therefore, (\ref{15}) is sufficient to guarantee that $P^{(n)}_{e,f^*} \rightarrow 0$ exponentially as $n\rightarrow \infty$. It is easy to check that $H_{\theta}(Z)= H_{\theta}(X|Y) = H_{\theta}(Y|X)$. The proof is complete.

\end{proof}


Theorem~\ref{thm:decodingZ} implies that the required rates to decode $Z^n = X^n \oplus Y^n$ with a small error probability is
\begin{eqnarray}
R_X \ge H_{\theta}(X|Y), \\
R_Y \ge H_{\theta}(Y|X).
\end{eqnarray}

This rate region is larger than the Slepian-Wolf region in~\eqref{eq:ratex}-\eqref{eq:ratexy}, as the condition $R_X+R_Y \ge H_{\theta}(X,Y)$ is not necessary anymore.

\subsection{Step 2: Estimation} \label{sub:estimation}
After obtaining $\hat{Z}^n$, which is equal to $Z^n$ with a probability converging to 1 exponentially, we then design 
an asymptotically MVUE of $\theta$. Our estimator is
\begin{eqnarray}
\hat{\theta}=\frac{n(0|\hat{Z}^n)}{n},\label{eq:estimator}
\end{eqnarray}
in which the notation $n(\cdot|\cdot)$ is defined in Section~\ref{sec:formulation}.

\begin{theorem}\label{thm:estimate}
If the conditions in Theorem~\ref{thm:decodingZ} are satisfied, the estimator in~\eqref{eq:estimator} is an asymptotically MVUE and achieves the optimal variance index as that of the centralized case.
\end{theorem}

\begin{proof}

Consider the centralized case in which $X^n$ and $Y^n$ are both known perfectly. Let $\left(\frac {n_1} n,\frac {n_2} n,\frac {n_3} n,\frac {n_4} n\right)$ denote the joint type of the sequences $x^n$ and $y^n$, where $(n_1,n_2,n_3,n_4)$ are the frequencies of occurrence of the pairs $\{(0,0),(1,1),(0,1),(1,0)\}$, respectively. The joint PMF of $(x^n,y^n)$ is
\begin{eqnarray}
P_\theta(x^n,y^n)&=& \bigg(\frac \theta 2\bigg)^{(n_1+n_2)}\bigg(\frac {1-\theta} 2\bigg)^{(n_3+n_4)} \label{likelihood}
\end{eqnarray}

Consider the centralized estimator
\begin{equation}
\hat \theta_c= \frac {(n_1+n_2)} n \label{estim}.
\end{equation}

This estimator is unbiased since
\begin{equation}
\mathbb E_{\theta}[\hat \theta_c]= \theta.
\end{equation}

The variance of the estimator is calculated as
\begin{eqnarray}
\text{Var}_\theta[\hat \theta_c] &=& \frac 1 {n^2} \mathbb E_{\theta}[(n_1+n_2)^2] - \theta^2 \nonumber \\
&=& \frac {\theta(1-\theta)} n.
\end{eqnarray}

The variance index is given by
\begin{eqnarray}
V[\hat \theta_c] = \lim_{n \rightarrow \infty} n\text{Var}_\theta[\hat \theta_c] = \theta(1-\theta).
\end{eqnarray}

The Cramer-Rao lower bound (CRLB) of the centralized case 
\begin{eqnarray}
\text{CRLB}&=&-1/\mathbb E_{\theta}\bigg[\frac{\partial^2 \ln [P_\theta(x^n,y^n)]}{\partial^2 \theta}\bigg]\\
&=& \frac {\theta(1-\theta)} n=\text{Var}_\theta[\hat \theta_c].
\end{eqnarray}
This implies that $\hat{\theta}_c$ is an MVUE for the centralized case.

Now, come back to our decentralized case.
For our estimator
\begin{equation}
\hat \theta= \frac {n(0|\hat{Z}^n)} {n},
\end{equation}
we have that
\begin{equation}
P_r(n(0|\hat{Z}^n)=n_1+n_2) \ge 1-P_{e,f^*}^{(n)}, \label{45}
\end{equation}
in which $P_{e,f^*}^{(n)}$ is shown to converge to zero exponentially fast in Section~\ref{sub:encoding}.

Therefore
\begin{eqnarray}
&&\mathbb E_{\theta}[\hat \theta] = P_r(n(0|\hat{Z}^n)=n_1+n_2)\mathbb E_{\theta}[\hat \theta_c] \nonumber \\
&&\hspace{10mm}+(1-P_r(n(0|\hat{Z}^n)=n_1+n_2))K_1,
\end{eqnarray}
where $K_1 \in [0,1]$ is a constant. As $n \rightarrow \infty$, $P_{e,f^*}^{(n)}\rightarrow 0$ and hence $P_r(n(0|\hat{Z}^n)=n_1+n_2) \rightarrow 1$. Therefore,
\begin{equation}
\lim_{n \rightarrow \infty} \mathbb E_{\theta}[\hat \theta] = \mathbb E_{\theta}[\hat \theta_c]=\theta.
\end{equation}
This shows that our estimator is asymptotically unbiased. Similarly,
\begin{eqnarray}
&&\hspace{-7mm}V[\hat \theta]= \lim_{n \rightarrow \infty} n\text{Var}_\theta[\hat \theta] \nonumber \\
&&= \lim_{n \rightarrow \infty} (nP_r(n(0|\hat{Z}^n)=n_1+n_2)(\mathbb E_{\theta}[\hat \theta_c^2]-(\mathbb E_{\theta}[\hat \theta_c])^2) \nonumber \\
&&\hspace{10mm}+n(1-P_r(n(0|\hat{Z}^n)=n_1+n_2))(K_2-K_1^2)), \nonumber
\end{eqnarray}
where $K_2 \in [0,1]$ is a constant. As $n \rightarrow \infty$, $P_{e,f^*}^{(n)} \rightarrow 0$ exponentially. Therefore,
\begin{equation} \label{25}
V[\hat \theta] = \theta(1-\theta)=V[\hat{\theta}_c].
\end{equation}
This proves that our estimator is asymptotically unbiased and achieves the same minimum variance that can be achieved even in the centralized case. Hence, our estimator is optimal.
\end{proof}

\begin{figure}[htb]
\centering
\includegraphics[width=0.34\textwidth]{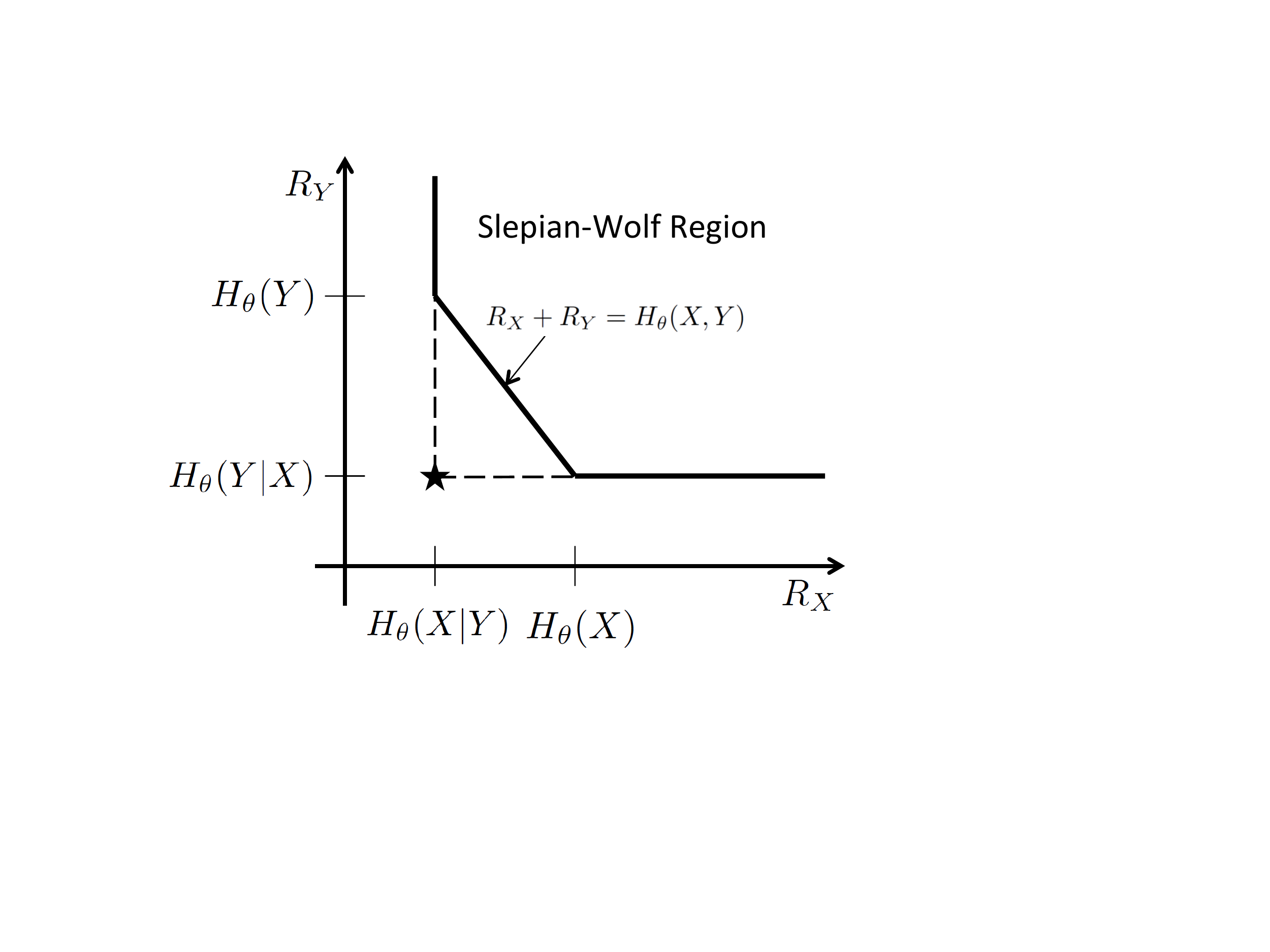}\\
\caption{$\bigstar$: the rate pair required in our estimator, which is outside of the Slepian-Wolf rate region.}
\label{Fig-2}
\end{figure}

Combining Theorems~\ref{thm:decodingZ} and~\ref{thm:estimate}, we conclude that, in the distributed parameter estimation, the Slepian-Wolf rates are not necessary to achieve the same optimal estimation performance as that of the centralized case. Fig.~\ref{Fig-2} illustrates the comparison between the Slepian-Wolf rate region and the rate pair used in our estimator.

\section{Extension} \label{sec:general}
In this section, we extend our results obtained in Section~\ref{sec:binary} to a more general class of joint PMFs. Let $ \mathcal{X}=\mathcal{Y}=\{0,1,...,M-1\}$ and the class of PMFs be
\begin{eqnarray} \label{eq:general}
P_\theta(X=i,Y=j)=
\begin{cases}
\frac {\theta} M ~~\text{, if}~~(i+j) \not= M-1 \\
\frac {1-\theta(M-1)} M~~\text{, otherwise,}
\end{cases}
\end{eqnarray}
where $\theta \in \Theta = (0,\frac 1 {(M-1)})$. Notice that each information source has a uniform marginal PMF and setting $M=2$ recovers the binary case.



Similar to the binary case, we first use a linear code and minimum entropy decoder to reconstruct $Z^n=(X^n+Y^n)\mod M$ at the decoder and then design an estimator from $Z^n$. In this section, we use$\mod M$ to denote element-wise mod operation,

In particular, we use a linear code $f$ that maps $\{0,1,...,M-1\}^n$ to $\{0,1,...,M-1\}^k$. 
The encoded messages of the realizations $x^n \in \{0,1,...,M-1\}^n$ and $y^n \in \{0,1,...,M-1\}^n$ are
\begin{eqnarray}
f(x^n)&=&x^nA, \nonumber \\
f(y^n)&=&y^nA,
\end{eqnarray}
in which the code matrix $A$ has $n$ rows and $k$ columns with each entry taking values from $\{0,1,...,M-1\}$. The coding rate is
\begin{equation}
R= \frac k n \log {M}.
\end{equation}

The decoder first combines the encoded messages into a single message as
\begin{equation}
f(x^n) + f(y^n)\mod M.
\end{equation}


The final decoded message is given by
\begin{equation}
\hat z^n= \phi(f(z^n)),
\end{equation}
where $\phi$ the the minimum entropy decoding function. Following the same error probability analysis for the binary case, we can show that there exists a codebook $f^*$ (and hence a particular encoding matrix $A$) that achieves a probability of decoding error $P_{e,f^*}^{(n)} \rightarrow 0$ exponentially as $n \rightarrow \infty$ if
\begin{equation}
R \ge H_{\theta}(Z)=H_{\theta}(X|Y)=H_{\theta}(Y|X).
\end{equation}

Therefore, as long as
\begin{eqnarray}
R_X \ge H_{\theta}(X|Y), \label{eq:mrx}\\
R_Y \ge H_{\theta}(Y|X),\label{eq:mry}
\end{eqnarray}
we can reconstruct $Z^n=X^n+Y^n\mod M$ at the decoder with an exponentially diminishing error probability.

After obtaining $\hat{Z}^n$, which is equal to $Z^n$ with a probability converging to 1 exponentially, our estimator is
\begin{eqnarray}
\hat \theta= \frac {n-n(M-1|\hat{Z}^n)}{n(M-1)}.~\label{eq:Mestimator}
\end{eqnarray}

Following similar steps as those in the binary case, we can show that, if~\eqref{eq:mrx}-\eqref{eq:mry} are satisfied, the estimator in~\eqref{eq:Mestimator} is asymptotically unbiased and achieves a variance index
\begin{equation} \label{var}
V[\hat \theta] = \frac {\theta[1-\theta(M-1)]} {M-1}.
\end{equation}
We can further show that~\eqref{var} is the best variance index that can be achieved even in the centralized case. This implies that our scheme achieves the centralized performance using rates outside the Slepian-Wolf region.

\section{Simulation Results} \label{sec:simulation}

In this section, we compare our estimator to the best known estimator by Han and Amari~\cite{han}. 
In the simulation, we fix the unknown parameter $\theta$ and change the encoding rates $R_X$ and $R_Y$ such that
\begin{eqnarray}
R_X=R_Y=R\ge H_{\theta}(Z).
\end{eqnarray}




We conduct the comparison for $M=2$ and $M=4$ respectively.

For $M=2$, the variance index of our estimator is (\ref{25}),
while the variance index of the estimator by Han and Amari is calculated in example 3 of \cite{han}
\begin{eqnarray}
&&(\text{Var}_\theta[\hat \theta])_{HA} \simeq  \\ \nonumber &&\frac 1 {16a^2b^2}\bigg\{\frac 1 4-\bigg(\theta-\frac 1 2\bigg)^2
[1-(1-4a^2)(1-4b^2)]\bigg\},
\end{eqnarray}
where $a$ and $b$ are functions of $R_X$ and $R_Y$, whose expressions are given in (14.12) and (14.13) of \cite{han}, respectively.

\begin{figure}[htb]
\centering
\includegraphics[width=0.4\textwidth]{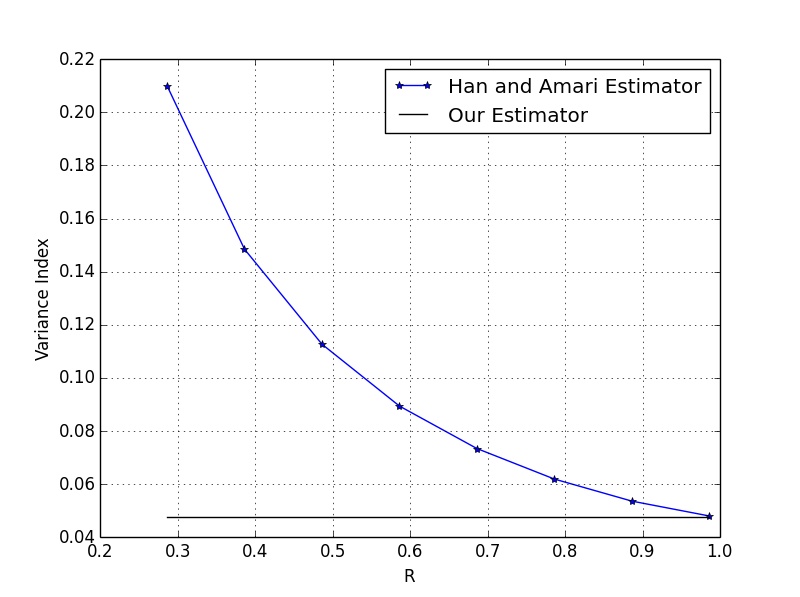}\\
\caption{Performance Comparison: $\theta=0.05$, $M=2$}
\label{Fig-3}
\end{figure}

\begin{figure}[htb]
\centering
\includegraphics[width=0.4\textwidth]{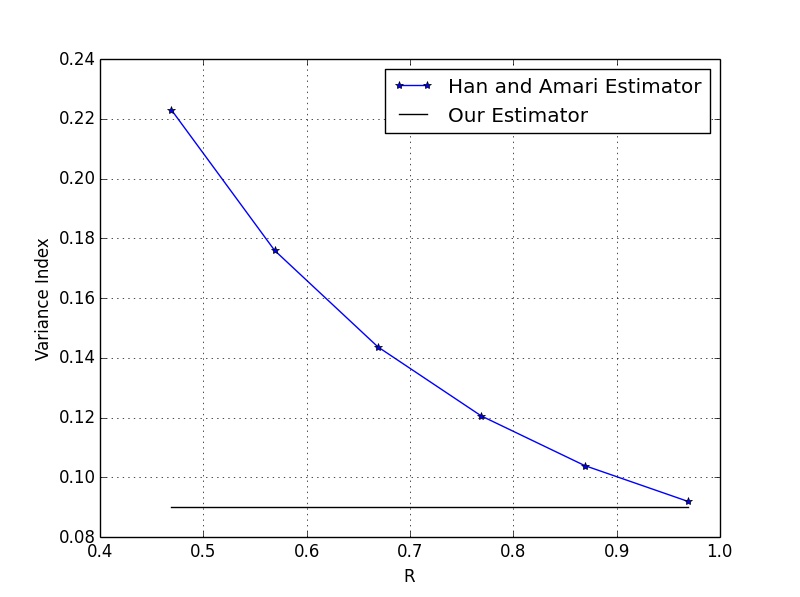}\\
\caption{Performance Comparison: $\theta=0.9$, $M=2$}
\label{Fig-4}
\end{figure}

Fig.~\ref{Fig-3} and Fig.~\ref{Fig-4} show the performance gain, in terms of the variance index, of our estimator over Han and Amari's estimator for binary symmetric sources ($M=2$) at two different values of the unknown parameter, $\theta=0.05$ and $\theta=0.9$, respectively. The performance difference is more noticeable at low rates. For $\theta=0.05$, the Slepian-Wolf sum rate is $R_X+R_Y=1.29$ bits, while our estimator requires a sum rate of $R_X+R_Y=2R=0.57$ bits. For $\theta=0.9$, the Slepian-Wolf sum rate is $1.47$ bits, while our estimator requires a sum rate of $0.94$ bits. Furthermore, for Han and Amari's estimator to achieve the centralized performance, the required sum-rate is $2$ bits for both cases, which is not only much larger than the sum rate required in our estimator but also much larger than the sum-rate required by conditions specified in the Slepian-Wolf rate region.

For $M=4$, the variance index of our estimator is given in (\ref{var}).
The performance of Han and Amari's estimator relies on the choice of the test channels. The authors did not specify an optimal choice of the test channels in order to extend example 3 in \cite{han} to the case of $M=4$. We find the following mapping to be a natural extension:
\begin{eqnarray}
Q=
\begin{cases}
0 \text{, if}~~X \in \{0,1\} \\
1 \text{, if}~~X \in \{2,3\},
\end{cases}
T=
\begin{cases}
0 \text{, if}~~Y \in \{0,1\} \\
1 \text{, if}~~Y \in \{2,3\}.
\end{cases}
\end{eqnarray}
Notice that $(Q,T)$ are distributed according to a binary symmetric PMF with an unknown parameter $\alpha = 2\theta$. Using an estimator $\hat \theta = \frac {\hat \alpha} 2$ leads to the following expression for the variance index:
\begin{eqnarray}
&&(\text{Var}_\theta[\hat \theta])_{HA} \simeq  \\ \nonumber&& \frac 1 {64a^2b^2}\bigg\{\frac 1 4-\bigg(2\theta-\frac 1 2\bigg)^2 [1-(1-4a^2)(1-4b^2)]\bigg\}.
\end{eqnarray}

\begin{figure}[htb]
\centering
\includegraphics[width=0.4\textwidth]{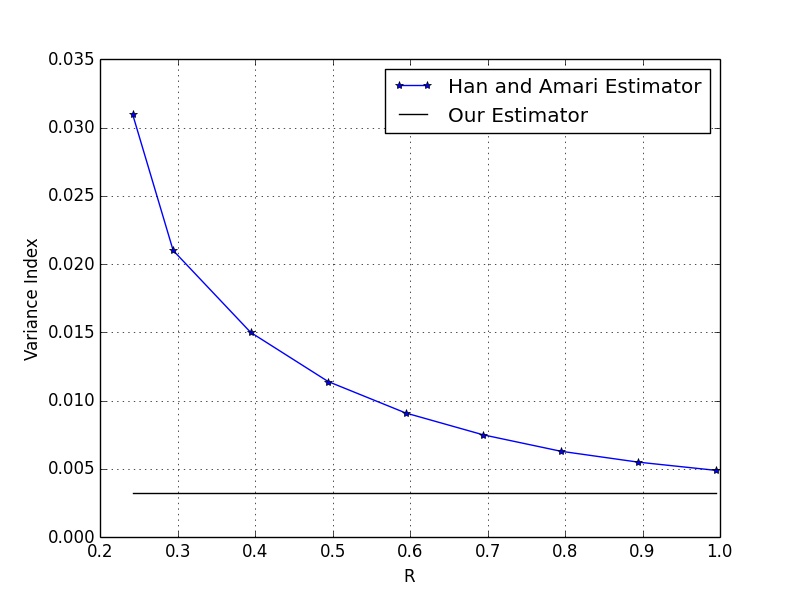}\\
\caption{Performance Comparison: $\theta=0.01$, $M=4$}
\label{Fig-5}
\end{figure}

Fig.~\ref{Fig-5} compares the variance indices achieved using our estimator and Han and Amari's estimator for $M=4$ and $\theta=0.01$. It is clear that our estimator outperforms that of Han and Amari's estimator. Furthermore, the performance difference is more noticeable at low rates. The Slepian-Wolf sum rate is $2.24$ bits, while our estimator requires a sum rate of $0.48$ bits.

\section{Conclusion} \label{sec:conclusion}
In this paper, we have answered the question: Are Slepian-Wolf rates necessary to achieve the same estimation performance as that of the centralized case? We have showed that the answer to this question is negative by constructing an asymptotically MVUE for binary symmetric sources using rates less than the conditions required in the Slepian-Wolf rate region. We have also extended our work to a general class of information sources by modifying the encoding/decoding scheme and the estimation algorithm. We have further compared our results to the best known estimator by Han and Amari to show the superiority of our estimator.
\bibliographystyle{ieeetr}
\bibliography{macros,allerton, }

\end{spacing}

\end{document}